\documentclass[a4paper, oneside]{amsart}
\usepackage[a4paper]{geometry}

\usepackage[foot]{amsaddr}
\usepackage{amsmath}
\usepackage{amssymb}
\usepackage{amsthm}
\usepackage{stmaryrd}
\usepackage[backref=page]{hyperref}

\usepackage{xcolor}
\usepackage{bm}

\usepackage{complexity}

\hypersetup{colorlinks = true, citecolor = blue, linkcolor = blue}

\title{Counting perfect matchings and Hamiltonian cycles faster}

\author{Baitian Li}
\email{bl3052@columbia.edu}
\address{Columbia University}

\newtheorem{theorem}{Theorem}
\newtheorem{lemma}{Lemma}
\newtheorem{corollary}{Corollary}

\theoremstyle{definition}
\newtheorem{definition}{Definition}

\let\poly\relax

\DeclareMathOperator{\poly}{poly}
\DeclareMathOperator{\per}{per}
\DeclareMathOperator{\haf}{haf}
\DeclareMathOperator{\hc}{hc}
\DeclareMathOperator{\sgn}{sgn}
\DeclareMathOperator{\inv}{inv}
\DeclareMathOperator{\rep}{rep}

\newcommand{\bbF}{\mathbb F}

\newcommand{\permanent}{\textsc{Permanent}}
\newcommand{\hamcycle}{\textsc{HamCycles}}
\newcommand{\hafnian}{\textsc{Hafnian}}
\newcommand{\iv}[1]{\llbracket {#1} \rrbracket}

\renewcommand{\emptyset}{\varnothing}
\renewcommand{\setminus}{\smallsetminus}

\DeclareMathOperator{\HH}{H}

\begin{document}

\begin{abstract}
    We show that the hafnian of a symmetric $2n\times 2n$ matrix
    of $\poly(n)$-bit integers (which counts the number of perfect matchings
    of a $2n$-vertex graph)
    and the number of Hamiltonian cycles of an $n$-vertex directed graph
    can be computed in time $2^{n-\Omega(\sqrt{n})}$,
    improving and generalizing an earlier algorithm of Bj\"orklund, Kaski, and Williams (Algorithmica 2019)
    that runs in time $2^{n - \Omega\left(\sqrt{n/\log \log n}\right)}$.

    A key tool of our approach is the design of a data structure that supports fast
    evaluation of high-order derivatives of hafnian and Hamiltonian cycles, which
    integrates with the new approach
    on multivariate multipoint evaluation by Bhargava, Ghosh, Guo, Kumar, and Umans (FOCS 2022, JACM 2024).
\end{abstract}

\maketitle

\section{Introduction}

Given an $n \times n$ matrix $A$ over a commutative ring $R$, the {$R$-\permanent} is defined by
\[
    \per A = \sum_{\sigma \in S_n} \prod_{i=1}^n A_{i, \sigma(i)},
\]
where $S_n$ denotes the symmetric group on $[n]$, i.e., permutations of $\{1,\dots,n\}$. Similarly, {$R$-\hamcycle} is defined as
\[
    \hc A = \sum_{\substack{\sigma \in S_n \\ c(\sigma) = 1}} \prod_{i=1}^n A_{i, \sigma(i)},
\]
where $c(\sigma)$ denotes the number of cycles in $\sigma$.

The permanent and Hamiltonian cycles are two fundamental problems in computer science.
The problem of deciding whether a given graph has a Hamiltonian cycle is one of Karp's 21 \NP-complete
problems \cite{Karp72NPC}.
Valiant proved that over the integers, computing the permanent is $\#\P$-complete,
even if the entries of the matrix are restricted to $0$ and $1$ \cite{Valiant79Per},
and counting Hamiltonian cycles is also $\#\P$-complete \cite{Valiant79Hamil}.

Ryser's formula \cite{Ryser63Per} shows that the permanent can be computed with $O(n2^n)$
arithmetic operations. It remains a prominent open problem whether the permanent can be computed
with arithmetic circuits of size less than $2^n$, as mentioned by Knuth  \cite[Exercise 4.6.4.11]{TAOCPVol2}.

Indeed, beyond the confines of arithmetic operations, faster algorithms for computing the permanent have emerged.
Bax and Franklin \cite{BF96Per} gave an algorithm that computes the $01$-permanent in $2^{n-\Omega(n^{1/3}/\log n)}$
expected time.
For dense instances over finite fields and integers, Bj\"orklund \cite{B16Per} introduced a framework based on
self-reduction and tabulation, achieving a running time of $2^{n-\Omega\left(\sqrt{n/\log n}\right)}$.
Bj\"orklund, Kaski, and Williams \cite{BKW19Per} refined the tabulation step via Kakeya sets, obtaining
an improved running time $2^{n-\Omega\left(\sqrt{n/\log \log n}\right)}$.

\subsection{Our results}

In this paper, we further improve the algorithm of Bj\"orklund, Kaski, and Williams \cite{BKW19Per},
removing the $\log \log n$ term in the exponent. We also show how to extend the complexity bound
for permanent to a natural extension called hafnian.

For a $2n\times 2n$ symmetric matrix $A$ over a commutative ring $R$,
the {$R$-\hafnian} of $A$ is defined as
\[
    \haf(A) = \sum_{\sigma \in P_{2n}} \prod_{(i, j)\in \sigma} A_{i,j},
\]
where $P_{2n}$ is the family of partitions of $[2n]$ into $n$ pairs.
The permanent of an $n\times n$ matrix can be reduced to the hafnian of a $2n\times 2n$ matrix
via the following basic relation:
\[
    \per(A) = \haf \begin{pmatrix} 0 & A \\ A^{\mathsf T} & 0 \end{pmatrix}.
\]

Previously the best known algorithm for hafnian ran in time $\tilde O(2^n)$ (for a $2n\times 2n$ matrix),
first developed by Bj\"orklund \cite{bjo12haf} and alternatively by Cygan and Pilipczuk \cite{CP15PerAvgDeg}.

\begin{theorem} \label{thm:finitefield}
    There is an algorithm that computes the permanent $\per (A)$
    of a given matrix $A \in \bbF_q^{n \times n}$
    in time $2^{n - \Omega(\sqrt{n})}q^{O(1)}$.
    The same bound holds for hafnian $\haf (A)$
    of a given symmetric matrix $A \in \bbF_q^{2n \times 2n}$,
    and for computing Hamiltonian cycles $\hc(A)$
    of a given matrix $A \in \bbF_q^{n\times n}$.
\end{theorem}

The Chinese remainder theorem and a simple estimate of prime products
yield the following corollary for integer-valued matrices.

\begin{corollary} \label{cor:integer}
    Given a $2n \times 2n$ symmetric matrix with integer entries having absolute values bounded by $M$,
    we can compute $\haf(A)$ (or $\per(A)$) in time $2^{n - \Omega(\sqrt{n})}(\log M)^{O(1)}$.
    The same type of bound holds for computing Hamiltonian cycles $\hc(A)$
    of a given $n\times n$ integer matrix.
\end{corollary}

\subsection{Related works}

\emph{Multivariate Multipoint Evaluation.}
Our algorithm is inspired by progress in multivariate multipoint evaluation.
Kedlaya and Umans \cite{KU11PolyFact} introduced a tabulation-based approach (combined with the Chinese remainder theorem)
that later became a key ingredient in fast polynomial composition and factorization.
More recently, a sequence of works \cite{BGKMK22MultiEval, BGGKU22MultiEval, BGGKU24MultiEval} developed the use of
Hasse derivatives and Hermite interpolation to extract more information per evaluation point. We adapt these ideas
to our sparse tabulation framework.

\emph{Permanents.}
There exist faster algorithms for computing the permanent in other settings.
For sparse matrices, Cygan and Pilipczuk \cite{CP15PerAvgDeg} gave a $2^{n - \Omega(n/d)}$ time algorithm,
where $d$ is the average degree of non-zero entries per row. Bj\"orklund and Williams \cite{BW19PerDissVec}
gave a $2^{n - \Omega(n/d^{3/4})}$ time algorithm for $d$-regular bipartite graphs, and a
$2^{n - \Omega(n/r)}$ time algorithm that runs over a finite ring with $r$ elements.
Bj\"orklund, Husfeldt, and Lyckberg \cite{BHL17PerModP} gave a $2^{n - \Omega(n/(p \log p))}$ time algorithm
for computing the permanent modulo a prime power $p^{\lfloor \lambda n / p\rfloor}$, for any constant $\lambda < 1$.

\emph{Hamiltonian cycles.}
There exist faster algorithms for counting Hamiltonian cycles in other settings as well. Bj\"orklund, Kaski, and Koutis \cite{BKK17HamilModP} gave an $O((2-\delta)^n)$-time algorithm for counting Hamiltonian cycles modulo moderate prime powers. In the general setting, it is somewhat surprising that our \emph{counting} algorithm also yields the fastest known algorithm for \emph{deciding} Hamiltonicity. This differs from the case of permanents: the support of the permanent corresponds to perfect matchings in bipartite graphs, whose existence can be decided in polynomial time. Faster decision algorithms are known in special cases, including Bj\"orklund's $O(1.66^n)$-time algorithm for undirected graphs \cite{Bjo14unHam} and the $O(1.888^n)$-time algorithm of Cygan, Kratsch, and Nederlof for directed bipartite graphs \cite{CKN18BipHamil}.

\subsection{Technical overview}

For simplicity, we first sketch the case of computing the permanent.

Our improvement comes from combining three ideas:

\begin{enumerate}
    \item \emph{Reduce to smaller instances.}
    We reduce the computation on an $n\times n$ matrix to many instances on $k\times k$ matrices.
    Taking $k$ around $\sqrt n$ is what creates room for an improvement in the exponent, provided we can answer the reduced instances fast.
    This kind of reduction (``self-reduction'') for permanents and Hamiltonian cycles was introduced by Bj\"orklund \cite{B16Per}.
    More concretely, for a parameter $k$, the reduction produces about $2^{n-k}\poly(n)$ instances on $k\times k$ matrices.

    \item \emph{Tabulate only on a sparse set of points.}
    A direct lookup table for all smaller matrices would be far too large. Instead, we employ the fact that permanent is a low-degree polynomial---we tabulate only on a carefully chosen
    sparse subset of points with the following key property: for every query point, there exists a low-degree univariate curve passing through it
    whose other points all lie inside the tabulated subset. Then we can recover the value at the query point by interpolating along that curve.
    (Over finite fields, such subsets are called \emph{Kakeya sets}. Bj\"orklund, Kaski,
    and Williams \cite{BKW19Per} were the first to leverage this idea for multivariate polynomial multipoint evaluation in our setting, and we
    build on their approach.)

    \item \emph{Make each tabulated point more informative.}
    Interpolating from plain point evaluations is limited by how many points we see on a curve.
    We use the recent idea from Bhargava, Ghosh, Guo, Kumar, and Umans \cite{BGGKU22MultiEval, BGGKU24MultiEval} to enrich each tabulated point
    with additional local information (captured via suitable high-order derivatives), and then use Hermite interpolation to recover higher-degree
    information along the curve. This is what allows the sparse tabulation approach to only require smaller Kakeya sets, thus working at the larger subproblem sizes we need.
    (We refer to this task as \emph{high-order derivative evaluation}.)
\end{enumerate}

Our main technical contribution is a dynamic programming algorithm that makes the required derivative access efficient for the permanent,
and we develop analogous data structures for hafnian and Hamiltonian cycles.
More concretely, for any constant $\epsilon > 0$ and a parameter $k$, our data structure
tabulates over a Kakeya set of size $2^{O_\epsilon(k^2)}$\footnote{Here $O_\epsilon(\cdot)$ means that the hidden constant factor depends on $\epsilon$.} and takes $O(2^{\epsilon k})$
time to evaluate one $k\times k$ permanent,
so the total time complexity is $O(2^{n-\Omega(k)} + 2^{O(k^2)})$.
Balancing the savings from self-reduction with the costs of
tabulation yields the final running time $2^{n - \Omega(\sqrt{n})}$.

The same high-level framework extends to hafnian and Hamiltonian cycles. For hafnian, we derive a suitable self-reduction by
modifying components of Bj\"orklund's algorithm \cite{bjo12haf}. For Hamiltonian cycles, we design an efficient derivative-evaluation
data structure based on a determinant characterization \cite{BCKN15SingleExp}.

\subsection{Discussion}

With the tabulation of information on Kakeya sets in $k\times k$ dimensional space, our algorithm essentially
computes the hafnian and Hamiltonian cycles in $2^{n - \Omega(k)} n^{O(1)}$ time. It seems that a better construction of a smaller Kakeya set
of size $2^{o(k^2)}$ would lead to a faster algorithm. However, the resolution of the finite field Kakeya conjecture
\cite{dvir2009kakeya, DKSS13Multi} rules out such possibilities, showing that the size of a Kakeya set is at least
$\Omega(\delta^{k^2})$ when the degree of the curve is not greater than $q/\delta$, which corresponds to the regime
of our application. Thus, the current construction is 
essentially optimal for our purposes.

It seems that we have reached a limit with the current approach of self-reduction and Kakeya sets. It remains open whether the
techniques on perturbing Ryser's formula, which work well for sparse permanents \cite{BW19PerDissVec} and
modulo $p^k$ permanents \cite{BHL17PerModP}, can be adapted to dense permanents to achieve a faster algorithm.

\section{Preliminaries}

\subsection{Notation}
We use $\tilde O(f(n))$ to denote $O(f(n)\poly \log(f(n)))$.

Bold symbols like $\bm{x}$ denote vectors $\bm{x} = (x_1,\dots,x_n)$.

For any positive integer $n$, we use $[n]$ to denote the set $\{1,2,\ldots,n\}$.

We use Iverson's bracket notation.
Let $P$ be a logical proposition, we let $\iv{P}$ be $1$ if $P$ is true and $0$ otherwise.

With $A\sqcup B$, we denote the disjoint union of two sets $A$ and $B$.

For an $n\times m$ matrix $A$, for subsets $S\subseteq[n]$ and $T\subseteq[m]$, we use $A_{S,T}$ to denote the
submatrix of $A$ with rows indexed by $S$ and columns indexed by $T$.

Let $\binom{n}{\downarrow m}$ denote the partial sum of binomial coefficients, i.e.,
\[
    \binom{n}{\downarrow m} = \sum_{0\leq i\leq m} \binom{n}{i}.
\]

\subsection{Inequality for binomials}

We need the estimate of the partial sum of binomials, see \cite[Lemma 3.13]{ExactExp} for a proof.

\begin{lemma} \label{lemma:binom}
    Consider $0 < \alpha < 1/2$. Then we have
    \[
        \binom{n}{\downarrow \alpha n} \leq 2^{n\HH(\alpha)},
    \]
    where $\HH(\alpha) = -\log_2(\alpha^\alpha (1-\alpha)^{1-\alpha})$ is the binary entropy function.
\end{lemma}

\subsection{Hermite interpolation}

We need the following lemma for Hermite interpolation, see \cite[Section 5.6]{ModernComputerAlgebra} for a proof.

\begin{lemma} \label{lemma:hermite}
    Let $f(t) \in \bbF[t]$ be a polynomial of degree less than $d$, and $m$ distinct points $\tau_1,\dots,\tau_m$ in $\bbF$,
    with multiplicities $e_1, \dots, e_m$ positive integers such that $e_1+\cdots+e_m = d$.
    Given the remainder polynomials $r_i = f \bmod (t - \tau_i)^{e_i}$ for each $i \in [m]$, then
    \begin{itemize}
        \item $f$ is uniquely determined by these $r_i$,
        \item moreover, the coefficients of $f$ can be recovered in $\poly(d)$ many $\bbF$-operations, given the coefficients of $r_i$ as input.
    \end{itemize}
\end{lemma}

In particular, our algorithm uses the case where those distinct points are the entire finite field $\bbF_q$,
and $e_i = r$ for all $i$.

\begin{corollary}
    Let $f(t)$ be a polynomial of degree less than $qr$. Given the coefficients of $f \bmod (t - \alpha)^r$ for all $\alpha\in \bbF_q$, then the coefficients of $f$ can be recovered in $\poly(qr)$ $\bbF_q$-operations.
\end{corollary}

\subsection{Multimodular reduction}

Our algorithm uses the Chinese remainder theorem to reduce the problem to small finite fields.

\begin{theorem} \label{thm:CRT}
    Let $p_1, \dots, p_n$ be distinct primes, and $a_1,\dots, a_n$ be integers
    such that $0\leq a_i < p_i$. Let $M = p_1\cdots p_n$. Then there exists a unique integer $a$
    in the range $0\leq a < M$ such that $a \equiv a_i \pmod{p_i}$ for every $i\in [n]$. Moreover, $a$ can be computed
    in time $\poly(\log M)$.
\end{theorem}

See \cite[Section 10.3]{ModernComputerAlgebra} for a proof.

We also need an estimate on the product of primes.
\begin{lemma} \label{lemma:primes}
    For an integer $N\geq 2$, we have
    \[ \prod_{\substack{\mathrm{prime~}p\\ p\leq 16\log N}} p > N. \]
\end{lemma}
See \cite[Lemma 2.4]{KU11PolyFact} for a proof.

\section{Common framework}

In this section, we set up the common framework for computing hafnians and counting
Hamiltonian cycles.

\subsection{Self reduction}

We borrow the self-reduction lemma of Hamiltonian cycles from \cite[Lemma 4]{B16Per}.

\begin{lemma}
    Suppose $|\bbF|\geq k^2 + 1$, given a matrix $A\in \bbF^{n\times n}$,
    one can compute $m = 2^{n-k}n^{O(1)}$ instances $a_i \in \bbF, F_i\in \bbF^{k\times k}$
    such that
    \[
        \hc(A) = \sum_{i=1}^m a_i \hc(F_i).
    \]
    Furthermore, the computation of these instances takes $2^{n-k}n^{O(1)}$ $\bbF$-operations.
\end{lemma}

\subsection{Kakeya set}

We borrow the definition and construction of Kakeya sets mentioned in \cite{BKW19Per}.

\begin{definition} \label{def:kakeya}
    A set $K\subseteq \bbF_q^m$ is said to be a \emph{Kakeya set} of degree $u$, if for every
    $a_1,\dots, a_m\in \bbF_q$, there exists degree-$u$ polynomials
    $g_1,\dots, g_m$, such that the degree $u$ coefficient of $g_i$ is $a_i$, and the set
    \[
        \{ (g_1(\tau), \dots, g_m(\tau)) : \tau \in \bbF_q \}
    \]
    is a subset of $K$.
\end{definition}

\begin{theorem} \label{thm:kaksize}
    Let $u$ be a positive integer such that $u+1$ divides $q-1$. Then there is a Kakeya set $K$
    of degree $u$ in $\bbF_q^m$ of size at most
    \[ \left(\frac{q-1}{u+1}+1\right)^{m+1}.\] Such $K$ can be
    constructed in time $|K| \cdot \poly(q)$ and for each point $\bm{a} = (a_1,\dots,a_m)\in \bbF_q^m$,
    the coefficients of the corresponding polynomials $g_1,\dots,g_m$ can be computed in time
    $\poly(u, m)$.
\end{theorem}

For the convenience of the reader, we provide the construction below, which is originally from \cite{MT04Construct}.
\begin{proof}
    Since $u+1$ divides $q-1$, we have $q \equiv 1 \pmod{u+1}$, thus $u+1$ is coprime with $q$,
    so $u+1$ is invertible in $\bbF_q$. For each point $\bm{a} = (a_1,\dots,a_m)\in \bbF_q^m$,
    we consider the polynomials given by
    \begin{align*}
        g_i(\tau) &= \left(\frac{a_i}{u+1} + \tau\right)^{u+1} - \tau^{u+1}\\
        &= \sum_{k=0}^u \binom{u+1}{k} \left(\frac{a_i}{u+1}\right)^{u-k+1} \tau^k 
    \end{align*}
    The $u+1$-th degree coefficient of $g_i$ cancels out, and the $u$-th degree coefficient is
    \[ \binom{u+1}{1} \frac{a_i}{u+1} = a_i. \]
    So these polynomials satisfy the leading monomial condition of Definition \ref{def:kakeya},
    and one can easily compute the coefficients of $g_i$ in polynomial time, through the explicit expression given above.

    Since $u+1$ divides $q-1$, from basic finite field theory, the set $T = \{ x^{u+1} : x\in \bbF_q \}$ consists of $0\in \bbF_q$ and roots of unity of order
    $d=(q-1)/(u+1)$, so $|T| = d+1$.
    Let $K$ be the set of points
    \[ (u_1 - v, \dots, u_m - v) : u_i \in T, v\in T. \]
    Since each $u_i$ and $v$ can take $d+1$ values, we have
    $|K|\leq (d+1)^{m+1}$. Moreover, for each $\tau\in \bbF_q$, we have
    $g_i(\tau) = u_i - v$ for $u_i = (a_i/(u+1) + \tau)^{u+1}$ and $v = \tau^{u+1}$.
    This shows that $K$ is indeed a Kakeya set of degree $u$, and satisfies the size bound.
    It is also straightforward to see that the construction can be done in time $|K|\cdot \poly(q)$.
\end{proof}

\subsection{High-order derivative evaluation}

\begin{definition}
    Let $P$ be a polynomial over $m$ indeterminates. We call the following operation a \emph{derivative evaluation}
    of $P$ at $\bm{a}\in \bbF_q^m$ up to \emph{order $r$} ($r$-order evaluation): Given a polynomial vector $\bm{f}(t) = (f_1(t),\dots, f_m(t))$, where each $f_i(t) \in \bbF_q[t]$ is a
    polynomial with degree less than $r$, and $\bm{f}(0) = \bm{a}$. Compute the coefficients of the polynomial $P(\bm{f}(t)) \bmod t^r$.
\end{definition}

This terminology comes from the intuition in characteristic zero.
In that case, computing $P(\bm{f}(t)) \bmod t^r$
is equivalent to computing all the derivatives of $P(\bm{f}(t))$ up to order $r$.

We rephrase the idea of \cite{BGGKU22MultiEval, BGGKU24MultiEval} to reveal information from derivatives.

\begin{theorem} \label{thm:reval}
    Let $P$ be a homogeneous degree $k$ polynomial over $m$ indeterminates, $b$ be a positive
    integer such that $q \equiv 1 \pmod b$. Let $u=(q-1)/b-1$ and $r = \lceil k/b\rceil$.
    Let $K$ be a Kakeya set of degree $u$, with an oracle that supports $r$-order evaluation query at any point of $K$.
    
    Then given any point $\bm{a}$ and the associated curve $\bm{C}_{\bm{a}}(t) = (g_1(t),\dots,g_m(t))$, we can compute $P(\bm{a})$ with $q$ oracle queries, and $\poly(k, q)$ arithmetic operations over $\bbF_q$.
\end{theorem}

\begin{proof}
    By the definition of Kakeya sets, it is guaranteed that $\bm{C}_{\bm{a}}(\tau)\in K$
    for all $\tau \in \bbF_q$.
    The polynomial $P(\bm{C}_{\bm{a}}(t))$ is of degree $ku$.
    Write $P(x_1,\dots,x_m)$ with
    \[
        P(x_1,\dots,x_m) = \sum_{\substack{i_1,\dots,i_m\in \mathbb N \\ i_1+\cdots+i_m = k}}
        p_{i_1,\dots,i_m} x_1^{i_1} \cdots x_m^{i_m},
    \]
    since $g_i(t) = a_i t^u + O(t^{u-1})$, we have
    \begin{align*}
        P(\bm{C}_{\bm{a}}(t)) &= \sum_{\substack{i_1,\dots,i_m\in \mathbb N \\ i_1+\cdots+i_m = k}}
        p_{i_1,\dots,i_m} g_1(t)^{i_1} \cdots g_m(t)^{i_m}\\
        &= \sum_{\substack{i_1,\dots,i_m\in \mathbb N \\ i_1+\cdots+i_m = k}}
        p_{i_1,\dots,i_m} (a_1 t^u + O(t^{u-1}))^{i_1} \cdots (a_m t^u + O(t^{u-1}))^{i_m}\\
        &=  \sum_{\substack{i_1,\dots,i_m\in \mathbb N \\ i_1+\cdots+i_m = k}}
        p_{i_1,\dots,i_m} (a_1^{i_1} \cdots a_m^{i_m} t^{ku} + O(t^{ku-1}))\\
        &= P(\bm{a})t^{ku} + O(t^{ku-1}),
    \end{align*}
    from which we have that the coefficient of $t^{ku}$ in $P(\bm{C}_{\bm{a}}(t))$ is $P(\bm{a})$.

    By the choice of $u$, we have $ku = k((q-1)/b-1) < qk/b \leq qr$.
    Let $Q(t) = P(\bm{C}_{\bm{a}}(t))$. If we are given $Q(t) \bmod (t-\tau)^{r}$ for each $\tau \in \bbF_q$, by
    Hermite interpolation (Lemma \ref{lemma:hermite}), we can recover $Q$ in
    $\poly(qr)$ operations. So the problem reduces to computing $P(\bm{C}_{\bm{a}}(t)) \bmod (t-\tau)^r$ for each $\tau \in \bbF_q$.
    
    In order to compute $Q(t) \bmod (t-\tau)^r$, one can write $Q(t) = R(t) + (t-\tau)^r D(t)$
    where $\deg R < r$, then $R(t)$ is the desired result. Thus we have
    $Q(t + \tau) = R(t + \tau) + t^r D(t+\tau)$, so we can compute $Q(t + \tau) \bmod t^r$, and then
    reveal $Q(t) \bmod (t-\tau)^r$ by substituting $t \gets t - \tau$. The conversion of
    coefficients only takes $\poly(r)$ arithmetic operations over $\bbF_q$. Thus we only need to compute
    $P(\bm{C}_{\bm{a}}(t + \tau)) \bmod t^r$ for each $\tau \in \bbF_q$. This is exactly an $r$-order evaluation
    of $P$ at $\bm{C}_{\bm{a}}(\tau)$.
\end{proof}

\section{Self reduction for hafnian}

In this section, we show that it is enough to take a truncation of Bj\"orklund's algorithm \cite{bjo12haf}
to obtain a self-reduction algorithm for the hafnian.

\subsection{Technical ingredients from Bj\"orklund's algorithm}
First, we list the technical ingredients we borrow from Bj\"orklund.

We first recap a basic concept introduced in \cite[Section 3.1]{bjo12haf}. For a commutative ring $R$, the \emph{set-partition algebra} $R[U_m]$ is defined as follows. Intuitively, $U_m$ may be viewed as a \emph{partial} semigroup whose elements are the subsets of $[m]$, with the operation given by disjoint union. Thus, the product of two subsets $U,V\subseteq [m]$ is defined precisely when $U$ and $V$ are disjoint. The algebra $R[U_m]$ is then the $R$-algebra associated with this partial semigroup, analogous to the group algebra $R[G]$ associated with a group $G$.

More explicitly, every element $r\in R[U_m]$ can be written uniquely as a formal sum
\[ r = \sum_{X \subseteq [m]} r_X [X]. \]
where $r_X\in R$. Addition in $R[U_m]$ is defined componentwise, while multiplication is given by
\[ r\cdot s = \sum_{X\subseteq[m]} \left(\sum_{Y \sqcup Z= X} r_Y s_Z\right) [X], \]
where the inner sum ranges over all ordered decompositions of $X$ as a disjoint union $Y\sqcup Z$.

We identify $R = R[U_0]$ and the inclusion
$R[U_m] \subseteq R[U_{m+1}]$ via the inclusion $[m] \subset [m+1]$.
Computationally, an element $r \in R[U_m]$ can be stored as the $2^m$ coefficients $r_X \in R$
where $X$ goes through all subsets of $[m]$. The multiplication in $R[U_m]$ is known as
the \emph{subset convolution}, which can be done in $\tilde O(2^m)$ $R$-operations \cite{bhkk07subset}.

Then, Bj\"orklund \cite[Section 3.2]{bjo12haf} introduced a sequence of matrices $B^{(i)} \in R[U_i]^{(2n-2i)\times (2n-2i)}$
(for $0\leq i\leq n$) starting with $B^{(0)} = A$, together with a sequence called \emph{squeeze factors}
$\beta^{(i)} \in R[U_i]$ (for $1\leq i\leq n$).
We will not need the precise definition of these matrices and factors, but
we invoke the following two facts.

\begin{lemma}[Bj\"orklund {\cite[Lemma 4]{bjo12haf}}] \label{lem:one-step}
    For every $X\subseteq [i-1]$ where $i\geq 1$, the following relation holds:
    \[ (\haf B^{(i-1)})_X = (\beta^{(i)} \cdot \haf B^{(i)})_{X \sqcup \{i\}}. \]
\end{lemma}

The above lemma has an efficient algorithmic counterpart.\footnote{Lemma \ref{lem:one-step-algo} is implicit in
Bj\"orklund's original paper.
In the original text Bj\"orklund merely stated how to sequentially compute $B^{(i)}$ and $\beta^{(i)}$
for all $1\leq i\leq n$ and gave a time complexity analysis, however, in the last paragraph
of \cite[Section 3.4]{bjo12haf}, the analysis of a single squeeze step is given.}
\begin{lemma}[Bj\"orklund {\cite[Section 3.4]{bjo12haf}}] \label{lem:one-step-algo}
    There is an algorithm that given $B^{(i-1)}$, computes $B^{(i)}$ and $\beta^{(i)}$ in
    $2^{i} n^{O(1)}$ $R$-operations.
\end{lemma}

\subsection{Self reduction via inclusion-exclusion}

We take several steps to obtain a self-reduction for hafnian.
For $1\leq \ell\leq n$, we denote $b^{(\ell)}$ to be the prefix product $b^{(\ell)} = \beta^{(1)}\cdots \beta^{(\ell)}$,
by repeatedly applying Lemma \ref{lem:one-step}, we obtain
\begin{equation} \label{eq:haf-b}
    \haf A = (\haf B^{(0)})_\emptyset = (b^{(\ell)} \cdot \haf B^{(\ell)})_{[\ell]}. 
\end{equation}

Then the next step involves extracting the ideas of ranked M\"obius transform and
inversion from the subset convolution algorithm \cite[Section 2]{bhkk07subset}.

We consider the ranked M\"obius transform, for which for every $X\subseteq [\ell]$,
we introduce the polynomial $\hat b^{(\ell)}_X \in R[T]$ defined as
\[
\hat b^{(\ell)}_X = \sum_{Y \subseteq X}  b^{(\ell)}_Y T^{|Y|}.
\]
Similarly, we define $\hat B^{(\ell)}_X \in R[T]^{(2n-2\ell)\times (2n-2\ell)}$ as
\[
\hat B^{(\ell)}_X = \sum_{Y \subseteq X} B^{(\ell)}_Y T^{|Y|}.
\]

For a polynomial $f\in R[T]$, let $[T^d] f$ denote the coefficient of $T^d$ in $f$.
We first prove a general statement and then apply it to equation \eqref{eq:haf-b}.
\begin{lemma}
    For a polynomial $F\in R[Y_1,\dots,Y_m]$
    and $y_1,\dots,y_m\in R[U_\ell]$, for every $X\subseteq [\ell]$, let
    \[ \hat y_{i,X} = \sum_{Y\subseteq X} y_{i,Y} T^{|Y|}, \]
    and $\hat f_X = F(\hat y_{1,X},\dots, \hat y_{m,X})$,
    let its M\"obius inversion be
    \[ f_X = \sum_{Y\subseteq X} (-1)^{|X|-|Y|} \hat f_Y, \]
    then we have
    \[ F(y_1,\dots,y_m)_X = [T^{|X|}] f_X. \]
\end{lemma}
\begin{proof}
    By linearity, it suffices to prove the case when $F$ is a monomial. Furthermore, we can
    without loss of generality, assume that $F$ is a monomial of the form
    \[ F(Y_1,\dots,Y_m) = Y_1 \cdots Y_m. \]
    In this case, we have
    \[ F(y_1,\dots,y_m)_X = (y_1\cdots y_m)_X = \sum_{X_1\sqcup \cdots \sqcup X_m = X} y_{1, X_1} \cdots y_{m, X_m}. \]
    On the other hand, since $\{ \hat y_{i, X} \}_X$ is the M\"obius transform of $\{ y_{i, X} T^{|X|} \}_X$,
    and $\{ f_X\}_X$ is the M\"obius inversion of $\{\hat f_X\}_X$,
    the basic property of the M\"obius transform gives us
    \[
        f_X = \sum_{X_1 \cup \cdots \cup X_m = X} y _{1, X_1} \cdots y_{m, X_m} T^{|X_1| + \cdots + |X_m|}.
    \]
    Since $X_1 \cup \cdots \cup X_m = X$, we have $|X_1| + \cdots + |X_m| \geq |X|$
    and the equality is attained when $X_1, \dots, X_m$ form a partition of $X$.
    Thus, we have
    \[
        [T^{|X|}] f_X = \sum_{X_1 \sqcup \cdots \sqcup X_m = X} y_{1, X_1} \cdots y_{m, X_m}.
        \qedhere
    \]
\end{proof}

In equation \eqref{eq:haf-b}, treating $b^{\ell} \cdot \haf B^{(\ell)}$ as a polynomial
in $b^{\ell}$ and all entries of $B^{(\ell)}$, we have the following corollary.
\begin{corollary} \label{cor:haf-b}
    For every $X\subseteq [\ell]$ define $\hat h^{(\ell)}_X, h^{(\ell)}_X$ by $\hat h^{(\ell)}_X = \hat b^{(\ell)}_X \cdot \haf (\hat B^{(\ell)}_X)$ and
    its M\"obius inversion
    \[ h^{(\ell)}_X = \sum_{Y \subseteq X} (-1)^{|X| - |Y|} \hat h^{(\ell)}_Y, \]
    then we have
    \[ \haf A = [T^\ell] h^{(\ell)}_{[\ell]}. \]
\end{corollary}

Now we are ready to give the self-reduction algorithm for hafnian.
\begin{theorem} \label{thm:haf-selfred}
    Let $\bbF_q$ be a finite field with $q \geq (k+1)(n-k)+1$. There is an algorithm that takes a symmetric matrix $A\in \bbF_q^{2n\times 2n}$
    as input, outputs $m = 2^{n-k}n^{O(1)}$ instances, consisting of $a_i \in \bbF_q$ and symmetric matrices $F_i \in \bbF_q^{2k\times 2k}$
    such that
    \[
        \haf (A) = \sum_{i=1}^m a_i \haf(F_i).
    \]
    This algorithm also runs in $2^{n-k}n^{O(1)}$ $\bbF_q$-operations.
\end{theorem}

\begin{proof}
    Consider the following algorithm.
    \begin{enumerate}
        \item Iteratively compute $B^{(i)}$ and $\beta^{(i)}$ for $1\leq i\leq n-k$ via Lemma \ref{lem:one-step-algo}.
        \item Compute $b^{(n-k)} = \beta^{(1)}\cdots \beta^{(n-k)}$ via fast subset convolution
        \cite{bhkk07subset}.
        \item Compute $\hat b^{(n-k)}_X$ and entries of $\hat B^{(n-k)}_X$ for all $X\subseteq [n-k]$,
        via fast M\"obius transform \cite[Section 2.2]{bhkk07subset}.
        \item Let $D=(k+1)(n-k)$. Since $q \geq D+1$, by Lagrange interpolation,
        it is possible to choose points $\alpha_0,\dots, \alpha_D\in \bbF_q$ and coefficients
        $\gamma_0,\dots, \gamma_D \in \bbF_q$ such that $[T^{n-k}]f(T) = \sum_{i=0}^{D} \gamma_i f(\alpha_i)$
        hold for every polynomial $f(T)$ of degree at most $D$.
        For every $X\subseteq [n-k]$ and $0\leq i \leq D$, compute the pair of a scalar and a matrix over $\bbF_q$:
        \[ \left( (-1)^{n-k-|X|} \gamma_i \hat b_X^{(n-k)}(\alpha_i)
        , \hat B_X^{(n-k)}(\alpha_i) \right) \]
        and they are the instances we need.
    \end{enumerate}
    We now explain the correctness of the algorithm.
    By Corollary \ref{cor:haf-b}, we have
    \begin{align*}
        \haf A &= [T^{n-k}] h_{[n-k]}^{(n-k)}\\
        &= [T^{n-k}] \sum_{X\subseteq [n-k]} (-1)^{n-k-|X|}\hat h^{(n-k)}_X\\
        &= \sum_{X\subseteq [n-k]} (-1)^{n-k-|X|} [T^{n-k}] \hat h^{(n-k)}_X\\
        &=  \sum_{X\subseteq [n-k]} (-1)^{n-k-|X|} [T^{n-k}] \hat b_X^{(n-k)} \cdot \haf(\hat B_X^{(n-k)}).
    \end{align*}
    Since $\haf (\hat B_X^{(n-k)})$ is a polynomial of degree $k$ in the entries of $\hat B_X^{(n-k)}$,
    $\hat b_X^{(n-k)} \cdot \haf(\hat B_X^{(n-k)})$ is a polynomial of degree $\leq D$ in $T$, 
    so we have
    \begin{align*}
        &\quad \sum_{\substack{X\subseteq [n-k] \\ 0\leq i \leq D}} (-1)^{n-k-|X|} \gamma_i \hat b_X^{(n-k)}(\alpha_i)
        \cdot  \haf(\hat B_X^{(n-k)}(\alpha_i))\\
        &= \sum_{X\subseteq [n-k]} (-1)^{n-k-|X|} [T^{n-k}] \hat b_X^{(n-k)}
        \cdot  \haf(\hat B_X^{(n-k)})
    \end{align*}
    correctly computes $\haf A$.
    Each step of the algorithm requires $2^{n-k}n^{O(1)}$ $\bbF_q$-operations, and outputs $2^{n-k} (D+1) = 2^{n-k}n^{O(1)}$ many instances, each is a tuple consisting of a scalar in $\bbF_q$ and a matrix in $\bbF_q^{2k\times 2k}$.
\end{proof}

\section{Data structure for hafnian}

\begin{lemma}
    For a commutative ring $R$ and symmetric matrices $A, B\in R^{2n\times 2n}$, we have
    \[ \haf (A+B) = \sum_{\substack{S\subseteq [2n] \\ |S| \equiv 0 \pmod 2}} \haf(B_{S,S})
    \haf(A_{[2n]\setminus S, [2n] \setminus S}). \]
\end{lemma}
\begin{proof}
    We give a combinatorial proof. The hafnian $\haf(A+B)$ takes the summation
    over perfect matchings of the complete graph $K_{2n}$ with the product of edge weights.
    By expanding the product of $(A+B)_{i,j}$, this is equivalent to coloring each selected edge
    with one of two colors $A$ and $B$, and taking the product of the weights of edges with the selected
    color. Hence we can first determine the vertices whose matching edges have color $A$ and $B$
    respectively, say vertices colored by $B$ form the set $S$ (whose cardinality must be even). Then the contribution
    of such a coloring is $\haf(B_{S,S}) \haf(A_{[2n]\setminus S, [2n] \setminus S})$.
\end{proof}

\begin{theorem} \label{thm:haf}
    Given a symmetric matrix $A\in \bbF^{2k\times 2k}$ and a positive integer $r$,
    we can precompute in time $\tilde O\left(\binom{2k}{\downarrow 2r} \cdot 2^k\right)$, and answer the $r$-order evaluation
    of hafnian at $A$ in time $\tilde O\left( \binom{2k}{\downarrow 2r} \right)$.
    Both are measured in $\bbF$-operations.
\end{theorem}
\begin{proof}
    We write $F(t) = A + B(t)$, where $B(t)$ has no constant term.
    
    Note that when $|S| \geq 2r$,
    the term $\haf(B_{S,S})$ does not contribute to the result. Let $f(S) = \haf(B_{S,S})$,
    these $f(S)$ can be computed via dynamic programming, described as follows.

    For the base case, we have $f(\emptyset) = 1$.

    For $0 < |S| < 2r$ and $|S| \equiv 0 \pmod 2$, let $s$ be a member of $S$,
    by enumerating the matching vertex $v$ of $s$, we have
    \[ f(S) = \sum_{v\in S \setminus \{s\}} B_{s, v} f(S \setminus \{s, v\}). \]
    After computing all $f(S)$, we can compute
    \[ \haf (A+B) = \sum_{\substack{S \subseteq [2n] \\ |S| \equiv 0 \pmod 2\\ |S| < 2r}}
    f(S) g_S, \]
    where $g_S = \haf(A_{[2n] \setminus S})$ can be precomputed via Bj\"orklund's algorithm
    in time $\tilde O(2^k)$.
    The precomputation time is $\binom{2k}{\downarrow 2r} \cdot \tilde O(2^k)
    = \tilde O\left(\binom{2k}{\downarrow 2r} \cdot 2^k\right)$, and each query
    takes time $\tilde O\left( \binom{2k}{\downarrow 2r} \right)$.
\end{proof}

Note that when $r=\alpha k$ for some $0 < \alpha < 1/2$, by Lemma \ref{lemma:binom}, precomputation takes
time $\tilde O(2^{(1 + 2\HH(\alpha))k})$, and each query takes time $\tilde O(2^{2\HH(\alpha) k})$.

\section{Data structure for Hamiltonian cycles}

In \cite{BCKN15SingleExp} they considered that Hamiltonian cycles can be counted as spanning
trees with restricted degree and used it to count undirected Hamiltonian
cycles in time exponential of treewidth. We give a directed version.

Let $\sigma\in S_n$ be a permutation, let $P_\sigma$ denote the permutation matrix associated with
$\sigma$, such that $(P_\sigma)_{i,j} = \iv{j = \sigma(i)}$.
\begin{lemma} \label{lemma:detham}
    For a permutation $\sigma\in S_n$,
    \[ \det \left((I - P_\sigma)_{[n]\setminus\{1\},[n]\setminus\{1\}}\right) = \iv{c(\sigma) = 1}. \]
\end{lemma}
\begin{proof}
    Consider a directed graph $G$ with directed edges $(i, \sigma(i))$, then $L = I - P_\sigma$ is exactly
    the Laplacian of the graph $G$. By the directed version of matrix tree theorem, $\det (L_{[n]\setminus\{1\},[n]\setminus\{1\}})$
    is the number of directed spanning trees rooted at vertex $1$. When $c(\sigma)=1$, then clearly
    there is exactly one spanning tree, otherwise there is no spanning tree.
    Thus we can conclude the claimed equality.
\end{proof}

Therefore, we use the above characterization of Hamiltonian cycles to help computing \hamcycle.

\begin{theorem} \label{thm:hamcycle}
    Given a matrix $A\in \bbF^{k\times k}$ and a positive integer $r$,
    we can precompute in $\tilde O\left( \binom{k}{\downarrow r} 4^k \right)$, and
    answer the $r$-order evaluation of Hamiltonian cycles polynomial at $A$ in time
    $\tilde O \left(\binom{k}{\downarrow r}^3 \right)$.
    Both are measured in $\bbF$-operations.
\end{theorem}

\begin{proof}
    By the definition of {\hamcycle} and Lemma \ref{lemma:detham}, we have
    \[
        \hc(A) = \sum_{\sigma\in S_k} \left(\prod_{i=1}^k A_{i,\sigma(i)}\right)
        \det \left((I - P_\sigma)_{[k]\setminus\{1\},[k]\setminus\{1\}}\right).
    \]
    We also expand the determinant by the Leibniz formula, i.e.,
    \[
        \det \left((I-P_\sigma)_{[k]\setminus\{1\},[k]\setminus\{1\}}\right) = \sum_{\substack{\tau\in S_k\\ \tau(1) = 1}}
        \sgn(\tau) \prod_{i=2}^k (I-P_\sigma)_{i, \tau(i)}.
    \]
    Combining the above two equations, and interpret $\sgn(\tau)$ as $(-1)^{\inv(\tau)}$, where
    $\inv(a)$ denotes the number of inversions for a sequence $a$, we have
    \begin{align*}
        \det \left((I-P_\sigma)_{[k]\setminus\{1\},[k]\setminus\{1\}}\right) &=
        \sum_{\substack{\sigma, \tau\in S_k\\ \tau(1) = 1}} \sgn(\tau)
        \left(\prod_{i=1}^k A_{i,\sigma(i)}\right) \left(\prod_{i=2}^k (I-P_\sigma)_{i, \tau(i)}\right)\\
        &= \sum_{\substack{\sigma, \tau\in S_k\\ \tau(1) = 1}} (-1)^{\inv(\tau)}
        \left(\prod_{i=1}^k A_{i,\sigma(i)}\right) \left(\prod_{i=2}^k \iv{i=\tau(i)} - \iv{\sigma(i) = \tau(i)}\right).
    \end{align*}

    Now consider dynamic programming. For $S \subseteq [k]\setminus \{1\}, T \subseteq [k]$
    and say $s = |S|=|T|$, let $f(S, T)$ only counts in the last $s$ values of $\sigma$ and $\tau$,
    with domain $\{\tau(k-s+1),\dots, \tau(k)\} = S$ and $\{\sigma(k-s+1),\dots, \sigma(k)\} = T$,
    and the inversions of $\tau$ in the last $s$ values are counted,
    i.e.,
    \begin{equation} \label{eqn:hcfinal}
        f(S, T) = \sum_{\sigma, \tau} (-1)^{\inv(\tau)} \left(\prod_{i=k-s+1}^k A_{i,\sigma(i)}\right)
        \left(\prod_{i=k-s+1}^k \iv{i=\tau(i)} - \iv{\sigma(i) = \tau(i)}\right).
    \end{equation}

    We let $a \gets^+ b$ denote $a\gets a+b$ for simplicity in describing the updating rules.
    The base case is simply $f(\emptyset, \emptyset) = 1$, and for each $s < k-1$, we use the computed values
    of $f(S, T)$ with $|S|=|T|=s$ to compute $f(S, T)$ with $|S|=|T|=s+1$ by the following rules.
    Let $i = k-s$.
    For each $j\notin T$, we can choose $\sigma(i)$ to be $j$, then there are
    two choices of $\tau(i)$:
    \begin{itemize}
        \item If $i\notin S$, update with
        \[
            f(S\cup \{i\}, T\cup \{j\}) \gets^+ (-1)^{\inv(i, S)}
            A_{i,j} f(S, T),
        \]
        denoting the choice that the contribution of term $\iv{i = \tau(i)}$ in equation \eqref{eqn:hcfinal}.
        \item If $j\notin S$, update with
        \[
            f(S\cup \{j\}, T\cup \{j\}) \gets^+ (-1)^{1+\inv(j, S)}
            A_{i,j} f(S, T),
        \]
        denoting the choice that the contribution of term $\iv{\sigma(i) = \tau(i)}$ in equation \eqref{eqn:hcfinal}.
    \end{itemize}
    Here $\inv(v, S)$ means the number of elements $x\in S$ such that $v > x$.

    Finally, we have the choice of $\sigma(1)$, thus
    \[
        \hc(A) = \sum_{i=1}^k A_{1, i} f([k] \setminus \{1\}, [k] \setminus \{i\}).
    \]

    This dynamic programming takes $\tilde O(4^k)$, which is slower than the usual
    one, but its dependence on the rows of $A$ is explicitly graded 
    by $s$, so is useful for our purpose.

    Now suppose the first $j$ rows are left undetermined, we can first preprocess
    all the $f(S, T)$ for $s \leq k - j$ in time $\tilde O(4^k)$, since their value does not depend
    on the first $j$ rows.
    Then for each query, i.e., given the first $j$ rows, can be computed in time
    \[
        \tilde O\left(\sum_{i=1}^j \binom{k}{i}\binom{k}{i-1} \right)
        = \tilde O\left( \binom{k}{\downarrow j}^2 \right).
    \]

    Write $F = A + B(t)$, where $B(t)$ has no constant term,
    by the multilinearity on rows of $\hc(\cdot)$, we have
    \[
        \hc(F(t)) = \hc(A+B) = \sum_{S\subseteq [k]} \hc(\rep_S(A, B)),
    \]
    where $\rep_S(A, B)$ denote the matrix obtained by replacing the rows indexed in $S$ of $A$ by
    those rows of $B$.
    The terms $|S| \geq r$ do not contribute to the result.
    For each $|S| < r$, we can reorder the rows and columns simultaneously to make
    $S$ be the first $|S|$ rows, and use the above dynamic programming to do precomputation
    and handle queries.
    
    There are $\binom{k}{\downarrow r}$ ways to choose $S$, so the precomputation needs
    $\tilde O\left(\binom{k}{\downarrow r} 4^k\right)$ time, and
    $\tilde O\left(\binom{k}{\downarrow r}^3\right)$ for each query.
\end{proof}

Note that when $r = \alpha k$ for some $0 < \alpha < 1/2$, by Lemma \ref{lemma:binom},
precomputation takes time $\tilde O(2^{(2 + \HH(\alpha))k})$, and each query takes time
$\tilde O(2^{3\HH(\alpha)k})$.

\section{The algorithms}

We first prove Theorem \ref{thm:finitefield} under some restrictions, and then remove the restrictions
by bootstrapping the results.

\begin{lemma} \label{lemma:haf}
    Let $q$ satisfy $q\geq n^2+1$ and $q \equiv 1 \pmod {b}$, where $b\geq 10$.
    There is an algorithm that computes the hafnian $\haf (A)$
    of a given symmetric matrix $A\in \bbF_q^{2n\times 2n}$
    in time $2^{n - \delta_b \sqrt n}q^{O(1)}$, for some $\delta_b > 0$.
\end{lemma}

\begin{proof}
    Let $\theta = \sqrt{\log(1.9)/\log(1+b)}$ and $k = \lfloor \theta \sqrt{n}\rfloor$, and consider the following algorithm.
    \begin{enumerate}
        \item First compute the Kakeya set by Theorem \ref{thm:kaksize} over $k^2$ variables of degree $u = (q-1)/b - 1$.
        \item Precompute the data structure for $r$-order evaluation for $r = \lceil k/b\rceil$
        at each point of $K$.
        \item Use the self-reduction of hafnian (Theorem \ref{thm:haf-selfred}) to reduce the problem to $m = 2^{n-k}n^{O(1)}$ instances
        of size $2k\times 2k$.
        \item For each instance, use Theorem \ref{thm:reval} to compute the hafnian.
    \end{enumerate}
    
    Then we analyze the time complexity. 
    In the precomputation phase,
    by Theorem \ref{thm:kaksize}, the size of Kakeya set is $(\frac{q-1}{u+1}+1)^{k^2+1} \leq (b+1)^{\theta^2 n+1} = O(1.9^n)$,
    and by Theorem \ref{thm:haf}, each data structure takes $2^{O(k)}$ time to precompute, so the total time of the first
    two steps is $1.9^{n+O(\sqrt n)} q^{O(1)}$.

    The data structure can answer $r$-order evaluation in time $\tilde O(2^{2\HH(\alpha)k})$.
    Here we have $\alpha = 1/b \leq 0.1$, hence $2\HH(\alpha) \leq 2\HH(0.1) < 0.94$, the total time
    in last two steps is
    \[
        2^{n-k}n^{O(1)} \cdot O(2^{0.94k}) q^{O(1)} = 2^{n-0.06k} q^{O(1)}
        = 2^{n-0.06 \theta\sqrt n} q^{O(1)}.
    \]
    In conclusion, we have $\delta_b = 0.06\theta$ satisfies the requirement.
\end{proof}

\begin{lemma} \label{lemma:hc}
    Let $q$ satisfy $q\geq n^2+1$ and $q \equiv 1 \pmod {b}$, where $b \geq 17$.
    There is an algorithm that computes Hamiltonian cycles $\hc (A)$
    of a given matrix $A\in \bbF_q^{n\times n}$
    in time $2^{n - \delta_b \sqrt n}q^{O(1)}$, for some $\delta_b > 0$.
\end{lemma}

\begin{proof}
    The algorithm is similar to the proof of Lemma \ref{lemma:haf}, with replacing
    the data structure for Hamiltonian cycles instead of hafnian.

    By Theorem \ref{thm:hamcycle}, the data structure can answer $r$-order evaluation of Hamiltonian cycles in time $\tilde O(2^{3\HH(\alpha)k})$,
    where $3\HH(\alpha) \leq 3\HH(1/17) < 0.97$.

    Then the total time in the last two steps is
    \[
        2^{n-k}n^{O(1)} \cdot O(2^{0.97k}) q^{O(1)} = 2^{n-0.03k} q^{O(1)}
        = 2^{n-0.03 \theta\sqrt n} q^{O(1)}.
    \]
    In conclusion, we have $\delta_b = 0.03\theta$ satisfies the requirement.
\end{proof}

\subsection{Proof of Theorem \ref{thm:finitefield}}
To prove Theorem \ref{thm:finitefield}, we only need to remove the conditions of
Lemma \ref{lemma:haf} and Lemma \ref{lemma:hc} on $q$ that $q\geq n^2+1$ and
$q\equiv 1 \pmod{b}$ for some fixed modulus $b$.

Note that for some integer $\ell$, we can embed $\bbF_q$ into a larger finite field $\bbF_{q^\ell}$. We only need to satisfy
$q^{\ell} \geq n^2+1$ and $q^{\ell} \equiv 1 \pmod{b}$. When $q$ is coprime with $b$, taking
$\ell = \varphi(b)$ is enough to satisfy the second condition, where $\varphi$ is the Euler totient function.
Taking $\ell$ as the smallest multiple of $\varphi(b)$ such that $q^{\ell} > n^2$, we have
$q^{\ell} \leq q^{\varphi(b)} n^2$.

For the hafnian, since $q$ is a prime power, it must be coprime with either $b=10$ or $b=11$.
For Hamiltonian cycles, $q$ must be coprime with either $b=17$ or $b=18$. Therefore,
we have $q^{\ell} = q^{O(1)} n^2$ since we only consider finite possibilities for $b$.

Therefore, by invoking the algorithms in Lemma \ref{lemma:haf} and Lemma \ref{lemma:hc}
through the finite field $\bbF_{q^\ell}$, we can compute the hafnian
and Hamiltonian cycles in time $2^{n - \Omega(\sqrt{n})}(q^\ell)^{O(1)} = 2^{n - \Omega(\sqrt{n})}q^{O(1)}$.

To actually support the computation in the finite field $\bbF_{q^\ell}$, we need to find an irreducible
polynomial $f$ and identify $\bbF_{q^\ell}$ as $\bbF_q[t]/(f)$. We can enumerate the polynomials
of degree $\ell$ over $\bbF_q$ and test whether they satisfy the conditions. By
\cite[Theorem 14.37]{ModernComputerAlgebra}, the time complexity of testing
irreducibility is $\poly(\ell, \log q)$. The time required to find an irreducible polynomial is
$O(q^{\ell} \poly(\log q^\ell))$, so this is not a bottleneck.
\hfill \qedsymbol

\subsection{Proof of Corollary \ref{cor:integer}}

The absolute value of $\haf(A)$ and $\hc(A)$ is trivially bounded by $C = (2n)!M^n$. Let $p_1,\dots,p_r$
be distinct prime numbers such that $D := \prod_i p_i > 2C + 1$. Then if we can compute
$\haf(A)$ and $\hc(A)$ modulo $D$, the values of $\haf(A)$ and $\hc(A)$ are uniquely determined.

By the Chinese remainder theorem, we only need to compute $\haf(A)$ and $\hc(A)$ modulo $p_i$ for each $i$,
and then combine them to get the result modulo $D$.

By Lemma \ref{lemma:primes}, the primes not greater than
$16\log D = O(n\log M)$ have their product greater than $D$. So we only need to compute $\haf(A)$ and $\hc(A)$
over finite fields $\bbF_q$ with $p = O(n\log M)$. By Theorem \ref{thm:finitefield}, we can compute
them in time $2^{n - \Omega(\sqrt{n})}p^{O(1)} = 2^{n - \Omega(\sqrt{n})}(\log M)^{O(1)}$.
There are $O(n\log M)$ instances to compute. Since the product of the chosen primes has $O(n\log M)$ bits,
by Theorem \ref{thm:CRT},
it takes $\poly(n\log M)$ time to combine them, which is not a bottleneck.
So the total time is $2^{n - \Omega(\sqrt{n})}(\log M)^{O(1)}$.
\hfill \qedsymbol

\section*{Acknowledgments}

The work was done when the author was an undergraduate student at Tsinghua University.
The author would like to thank Josh Alman and anonymous referees for helpful comments on earlier drafts.

\bibliographystyle{abbrvurl}
\bibliography{main.bib}

\end{document}